\documentclass{article}
\usepackage{graphicx,amsfonts,amsmath,mathrsfs,amssymb,amsthm,url,color,dsfont}
\usepackage{fancyhdr,indentfirst,bm,enumerate,natbib}
\usepackage[colorlinks=true,citecolor=blue]{hyperref}
                      
\def\oo{\infty}                   \def\d{\,\mathrm{d}}
\def\lm{\lambda}                  \def\th{\theta}
                 
\newcommand{\esssup}{\mathrm{ess\mbox{-}sup}}
 \newcommand{\essinf}{\mathrm{ess\mbox{-}inf}}

\newcommand{\E}{\mathbb{E}}
\newcommand{\R}{\mathbb{R}}
\newcommand{\N}{\mathbb{N}}

\newcommand{\I}{\mathbb{I}}
\renewcommand{\P}{\mathbb{P}}
\newcommand{\Q}{\mathbb{Q}}
\renewcommand{\(}{\left (}
\renewcommand{\)}{\right )}
\renewcommand{\[}{\left [}
\renewcommand{\]}{\right ]}
\newtheorem{theorem}{Theorem}[section]

\newtheorem{proposition}[theorem]{Proposition}
\newtheorem{definition}{Definition}[section]
\newtheorem{example}[theorem]{Example}

\newtheorem{remark}[theorem]{Remark}

\numberwithin{equation}{section}
\numberwithin{theorem}{section}

\renewcommand{\cite}{\citet}
\allowdisplaybreaks

\newcommand{\VaR}{\mathrm{VaR}}
\newcommand{\ES}{\mathrm{ES}}
\newcommand{\EVaR}{\mathrm{EVaR}}
\begin{document}
	
\title{Lambda R{\'e}nyi entropic value-at-risk}
	
\author{Zhenfeng Zou\thanks{School of Public Affairs, University of Science and Technology of China,  China. Email: \url{zfzou@ustc.edu.cn}}
}

\date{April, 2026}
\maketitle
	
\begin{abstract}

This paper introduces the Lambda extension of the R\'{e}nyi entropic value-at-risk ($\Lambda$-EVaR), a novel family of risk measures that unifies the flexible confidence level structure of the $\Lambda$-framework with the higher-moment sensitivity of EVaR. 
We define $\Lambda$-EVaR, establish its foundational properties including monotonicity, cash subadditivity, and quasi-convexity, and provide a complete axiomatic characterization showing that convexity, concavity in mixtures and cash additivity hold only when $\Lambda$ is constant. 
A dual representation and an extended Rockafellar-Uryasev-type formula are derived, enabling efficient computation. 
We further analyze the worst-case behavior of $\Lambda$-EVaR under Wasserstein and mean-variance uncertainty, obtaining closed-form expressions that reveal its robustness properties. 
The proposed measure bridges the gap between adaptive risk tolerance and moment-sensitive risk assessment, offering a versatile tool for modern risk management.

\medskip
		
\noindent \textbf{MSC2000 subject classification}: 91G70, 91B05
		
\end{abstract}

\section{Introduction}

Risk measures serve as the cornerstone for quantifying, managing, and regulating risk in finance, insurance, and decision theory. 
Among the plethora of proposed measures, Value-at-Risk (VaR) and Expected Shortfall (ES) have achieved paradigmatic status due to their intuitive interpretation and regulatory adoption (e.g., Basel Accords). 
Both are defined with respect to a fixed confidence level $\alpha \in [0,1]$, which determines the portion of the loss distribution considered. 
While this fixed-level paradigm offers simplicity, it has been widely criticized for its rigidity. 
A single $\alpha$ cannot reflect the nuanced risk tolerance of a decision-maker who may wish to be highly conservative for moderate, more frequent losses, yet more accepting of truly extreme, catastrophic events---a preference pattern often observed in practice and embedded in progressive regulatory frameworks.

To address this limitation, the innovative Lambda ($\Lambda$) risk measure framework was introduced. 
It replaces the constant $\alpha$ with a function $\Lambda: \R \to [0, 1]$, mapping each potential loss threshold $x$ to a local confidence level $\Lambda(x)$. 
This generalization allows the risk assessment to adapt dynamically to the loss size, formally encoding the intuition that ``how much risk we see depends on how bad the loss could be." 
The framework originated with $\Lambda$-VaR, where $\Lambda$ is taken as increasing, linking higher loss thresholds with higher confidence levels \citep{FMP14}. 
$\Lambda$-VaR has been shown to satisfy several properties that are useful in finance, including elicitability \citep{BB15}, robustness \citep{BPR17}, and quasi-star-shapeness \citep{HWWX25}.
\cite{BP22} further provided an axiomatic characterization of $\Lambda$-VaR, which particularly justifies the choice of $\Lambda$ as a (weakly) decreasing function. 
From a practical perspective, $\Lambda$-VaR has also been studied in various applied contexts, such as estimation and backtesting \citep{HMP18, CP18}, distributionally robust optimization \citep{HL26}, risk sharing \citep{Liu25, XH25}, and optimal insurance design \citep{BCHW25}.

Subsequently, the $\Lambda$ framework was recently extended from VaR to ES, employing a decreasing $\Lambda$ function by \cite{BHWW25}.
This modification aligns more naturally with economic and regulatory reasoning: for extreme tail losses (high $x$), insisting on an extremely high confidence level ($\Lambda(x) \approx 1$) may lead to prohibitive and unrealistic capital requirements; thus, allowing $\Lambda(x)$ to decrease for large $x$ provides necessary flexibility. 
Both $\Lambda$-VaR and $\Lambda$-ES preserve essential properties such as monotonicity, law invariance, and, notably, cash subadditivity---a relaxation of cash additivity pertinent in environments with stochastic or ambiguous interest rates \citep{ER09, CMMM11}.
Moreover, \cite{BHWW25} demonstrated that $\Lambda$-ES is the smallest quasi-convex and law-invariant risk measure dominating $\Lambda$-VaR, and provided both a dual representation and a Rockafellar-Uryasev-type formula \citep{RU02} for $\Lambda$-ES.

While $\Lambda$-VaR and $\Lambda$-ES offer valuable flexibility in confidence levels, they nonetheless inherit fundamental limitations from their classical counterparts.
VaR is not a convex risk measure, failing to incentivize diversification \citep{ADEH99}. 
ES, while coherent (convex, monotone, cash additive, and positively homogeneous), is a spectral risk measure that depends solely on the tail expectation, disregarding information about the shape of the loss distribution beyond the $\alpha$-quantile \cite{Acebi02}. 
In modern finance, where assets often exhibit heavy tails, skewness, and kurtosis, a risk measure sensitive to higher moments is crucial for accurate risk assessment and capital allocation. 
This raises a natural and pressing question: Can we construct a $\Lambda$-type risk measure that is both convex (or quasi-convex) and sensitive to higher-moment tail behavior?

The R\'{e}nyi Entropic Value-at-Risk (EVaR) emerges as an ideal candidate to answer this question. 
Introduced and further developed within a rigorous dual framework by \cite{PS20} and an adjusted version by \cite{ZWXH23}, EVaR is a coherent risk measure defined on $L^p$ spaces. 
For a fixed $\alpha$, $\EVaR_\alpha^p$ is derived from an optimization problem involving the R\'{e}nyi entropy of order $q$ (where $1/p+1/q=1$), effectively imposing a constraint on the divergence of a probability measure from the reference measure $\P$. 
Its primal representation admits a computationally convenient Rockafellar-Uryasev-type formula \citep{RU02}. 
Crucially, $\EVaR_\alpha^p$ dominates ES and incorporates information from the $p$-th moment of the loss tail. 
As $p$ increases, it interpolates between ES (at $p=1$) and more conservative, and higher-moment sensitive measures, making it a powerful and flexible tool for applications where tail dispersion matters \citep{PS20, DPR10}.

Building upon the success of the $\Lambda$-framework for VaR and ES, and motivated by the superior theoretical properties of EVaR, this paper introduces and systematically analyzes its Lambda extension: the Lambda R\'{e}nyi Entropic Value-at-Risk, denoted by $\EVaR^p_\Lambda$.
Our work synthesizes two active strands of research---flexible confidence levels and entropy-based risk measures---to create a unified and powerful risk assessment tool.
Our work advances the literature in three directions.

First, we define $\EVaR^p_\Lambda$ and investigate its axiomatic properties. 
We examine under what conditions it retains convexity, cash additivity, or positive homogeneity---key properties that distinguish classical risk measures from their Lambda counterparts. 
This analysis reveals how the flexibility introduced by a non-constant $\Lambda$ trades off with these classical axioms.
	
Second, we derive alternative representations that make $\EVaR^p_\Lambda$ analytically and computationally tractable. 
Specifically, we establish a dual representation in terms of R\'{e}nyi entropy and an extended Rockafellar-Uryasev formula that reduces its evaluation to a two dimensional optimization problem. 
These representations not only deepen the theoretical understanding but also facilitate practical implementation.

Third, we analyze the robustness of $\EVaR^p_\Lambda$ under model uncertainty. 
We obtain closed-form worst-case bounds when the underlying distribution is known only through moment constraints or lies within a Wasserstein ball. 
These results demonstrate how the Lambda structure interacts with distributional ambiguity, a question of both theoretical and practical importance in risk management.

Through these investigations, we show that $\EVaR^p_\Lambda$ provides a flexible yet mathematically sound framework that generalizes both $\Lambda$-VaR and $\Lambda$-ES while incorporating sensitivity to higher moments---a feature absent in existing Lambda risk measures.

The remainder of the paper is organized as follows. 
Section \ref{sec2} reviews necessary background on risk measures, the Lambda framework, and the classical R\'{e}nyi EVaR. 
Section \ref{sec3} introduces $\Lambda$-EVaR, establishes its basic properties, and characterizes when it reduces to a classical risk measure. 
Section \ref{sec4} presents dual representations, an extended Rockafellar-Uryasev formula, and worst-case results under Wasserstein and mean-variance uncertainty. 
Section \ref{sec5} concludes the paper.

\section{Preliminaries}\label{sec2}

Let $L^0$ denote the space of all random variables on an atomless probability space $(\Omega, \mathscr{F}, \P)$. 
For $p \geq 1$, let $L^p = L^p(\Omega, \mathscr{F}, \P)$ be the set of all random variables on $(\Omega, \mathscr{F}, \P)$ with finite $p$th-moment and let $L^\oo = L^\oo(\Omega, \mathscr{F}, \P)$ be the space of essentially bounded random variables. 
Throughout the paper, for $X \in L^0$, a positive [resp. negative] value of $X$ represents a financial loss [resp. profit]. 
We write $X \sim F$ to indicate that the distribution function of $X$ is $F$.  
All expectations and variances are computed with respect to $\P$ unless specified otherwise.
For two random variables $X, Y\in L^0$, we write $X\stackrel {\rm d}= Y$ when $X$ and $Y$ have the same distribution.
Denote $\R = (-\oo, \oo)$, $\overline{\R} = [-\oo, \oo]$ and $\R_+ = [0, \oo)$.
For any $x, y \in \overline{\R}$, write $x \wedge y = \max\{x, y\},\ x \vee y = \min\{x, y\},\ x_+ = x \vee 0$, and $x_- = x \wedge 0$. 
For a function $f : \R \to \R$ and a point $x \in \R$, we write $f(x^-) = \lim_{y \uparrow x} f(y)$ and $f(x^+) = \lim_{y \downarrow x} f(y)$, whenever these limits exist.
Let ${\cal M}_c(\R)$ denote the set of compactly supported distributions on $\R$.

\subsection{Basic terminology}

In this subsection, we briefly review the basic concepts and properties of risk measures used throughout the paper. 
For a thorough introduction to the theory of risk measures, we refer the reader to the monographs \cite{Delb12, FS16}.

A \textit{risk measure} is a functional $\rho: \mathcal{X} \to \overline{\R}$, where $\mathcal{X}$ is a suitable space of financial positions (typically $L^p$), which represents the minimal amount of capital that has to be raised and injected in a financial position $X$ to ensure that $X$ is acceptable (safe).
To characterize the behavior of a risk measure in a mathematically rigorous and economically meaningful way, several axiomatic properties have been widely adopted in the literature. 
The most commonly considered axioms for risk measures include the following: for $X, Y\in \cal X$,
\begin{itemize}
 \item[(A1)] Monotonicity: $\rho(X)\le \rho(Y)$ if $X \le Y$;
 \item[(A2)] Cash additivity: $\rho(X + m) = \rho(X) + m$ for all $m\in\R$;
 \item[(A3)] Subadditivity: $\rho(X + Y) \leq \rho(X) + \rho(Y)$;
 \item[(A4)] Positive homogeneity: $\rho(\lambda X) = \lambda\rho(X)$ for $\lm\ge 0$;
 \item[(A5)] Convexity: $\rho(\lambda X +(1 - \lambda) Y) \leq \lambda \rho(X) + (1 - \lambda) \rho(Y)$ for all $\lm\in [0,1]$;
 \item[(A6)] Law invariance: $\rho(X)=\rho(Y)$ if $X\stackrel {d}= Y$.
\end{itemize}
Note that the axiomatic properties listed above are not mutually independent; some may be implied by others depending on the chosen framework. 
Different selections of these properties lead to distinct classes of risk measures commonly studied in the literature. For instance:
\begin{itemize}
 \item A risk measure satisfying monotonicity (A1) and cash additivity (A2) is called a monetary risk measure.;
 \item If, in addition, it satisfies convexity (A5), it is referred to as a convex risk measure.;
 \item A coherent risk measure is a monetary risk measure that also fulfills positive homogeneity (A3) and positive homogeneity (A4).
\end{itemize}
In empirical and practical applications, it is essential that a risk measure can be estimated from observable data. 
This requirement is formally captured by the property of law invariance (A6), which ensures that the risk measure depends only on the distribution of the financial position. 
Throughout this paper, we restrict our attention to law invariant risk measures, as is standard in most applied and statistical settings.

To account for environments with non-constant interest rates, cash additivity has been relaxed to the weaker property of cash subadditivity. 
Specifically, a risk measure $\rho$ is called cash subadditive if $\rho(X + m) \leq \rho(X) + m$ for all $X \in \cal X$ and $m \in \R_+$ \citep{ER09}.
While cash subadditivity reflects the impact of varying discount rates, the diversification effect in this setting is characterized by quasi-convexity rather than convexity. 
That is, $\rho(\lambda X +(1 - \lambda) Y) \leq \max\{\rho(X), \rho(Y)\}$ for all $X, Y \in \cal X$ and $\lm \in [0, 1]$ \citep{CMMM11}.
More recently, \cite{HWWX25} studied the broader class of cash subadditive risk measures that satisfy only monotonicity and cash subadditivity, without requiring quasi-convexity, thereby extending the framework to accommodate a wider range of risk-assessment practices.

Finally, we revisit the definitions of concavity (and quasi-concavity, respectively) in the context of mixtures.
For a law invariant functional $\rho$ on $\cal X$, the mapping $F \mapsto \rho(X_F)$ on ${\cal M}_c(\R)$ is concave (resp. quasi-concave), where $X_F$ is a random variable with distribution $F \in {\cal M}_c(\R)$.

\subsection{Risk measures and their Lambda counterparts}\label{sec22}

The concept of Lambda risk measures generalizes traditional risk measures by allowing the confidence level to depend on the loss threshold. 
Instead of fixing a single confidence level $\alpha \in [0,1]$, one uses a function $\Lambda: \R \to [0,1]$, where $\Lambda(x)$ reflects the risk-tolerance at the loss level $x$. 
This idea was originally introduced for Value-at-Risk (VaR) by \cite{FMP14}, who considered an increasing function $\Lambda$, and was later extended to Expected Shortfall (ES) by \cite{BHWW25} who adopted a decreasing specification of $\Lambda$.
Here, as throughout the paper, the terms ``increasing" and ``decreasing" are understood in the non-strict (weak) sense.

Given a law invariant risk measure $\rho_\alpha$ indexed by a confidence level $\alpha \in [0,1]$, its Lambda counterpart $\rho_\Lambda$ is defined as
\begin{equation}
\rho_\Lambda(X) := \sup_{x \in \R} \left\{\rho_{\Lambda(x)}(X) \wedge x\right\} ~~~ X \in \cal X.
\end{equation}
Intuitively, for each threshold $x$ we compute the risk of $X$ at the ``local" confidence level $\Lambda(x)$, compare it with the threshold itself, and then take the supremum over all thresholds. 
The definition ensures that $\rho_\Lambda$ inherits many structural properties from the family $\{\rho_\alpha\}$, while gaining flexibility to adapt the confidence level to the size of the loss.

While the $\Lambda$-framework in principle allows for both increasing and decreasing specifications, the decreasing choice aligns naturally with economic intuition and regulatory practice. 
When the loss threshold $x$ is low, the decision-maker is typically concerned with relatively frequent, moderate losses and may therefore adopt a high confidence level (i.e., $\Lambda(x)$ close to 1). 
Conversely, for very high thresholds $x$ that correspond to extreme loss, a lower confidence level (i.e., $\Lambda(x)$ close to 0) is often more appropriate, as insisting on extremely high confidence may lead to unrealistically conservative capital requirements. 
This decreasing pattern mirrors the common regulatory philosophy that allows for more flexibility in the treatment of extreme tail risks.

Moreover, a decreasing $\Lambda$ guarantees that the resulting $\Lambda$-risk measure retains desirable properties such as elicitability and the axiomatic justification for $\Lambda$-VaR \citep{BB15, BP22}. 
For these reasons, we restrict our attention to decreasing $\Lambda$ throughout the paper, which is also the convention adopted in the recent literature on $\Lambda$-ES.

The two most prominent examples are Lambda VaR and Lambda ES.
Before we recall these two risk measures, we give the formal definitions for VaR and ES.
VaR and ES at confidence level $\alpha\in (0,1)$, are defined by
$$
   \VaR_\alpha(X)= \inf\{x:\ \P(X\le x)\ge \alpha\}
$$
and
\begin{equation*}
 \label{eq-1101}
   \ES_\alpha(X)=\frac {1}{1-\alpha} \int^1_\alpha {\rm VaR}_s(X) \d s,
\end{equation*}
respectively. In addition, let $\VaR_0(X)=\essinf(X)$, $\ES_0(X)=\E [X]$ and $\VaR_1(X)=\ES_1(X)=\esssup(X)$.
\begin{example}[$\Lambda$-VaR]{\rm 
For a decreasing function $\Lambda: \R \to [0,1]$,
\begin{equation}
\VaR_\Lambda(X) := \sup_{x \in \R} \left\{\VaR_{\Lambda(x)}(X) \wedge x\right\} ~~~ X \in L^0.
\end{equation}
Here $\VaR_{\Lambda(x)}$ is the usual VaR at level $\Lambda(x)$. 
The functional $\VaR_\Lambda$ satisfies monotonicity, law invariance, cash subadditivity, quasi-convex, and quasi-concave in mixtures but not cash additive or positively homogeneous. 
For more details basic properties, we refer to \cite{XH25}.
}
\end{example}

\begin{example}[$\Lambda$-ES]{\rm 
For a decreasing function $\Lambda: \R \to [0,1]$,
\begin{equation}
\ES_\Lambda(X) := \sup_{x \in \R} \left\{\ES_{\Lambda(x)}(X) \wedge x\right\} ~~~ X \in L^1.
\end{equation}
Here $\ES_{\Lambda(x)}$ is the usual ES at level $\Lambda(x)$. 
The functional $\ES_\Lambda$ satisfies monotonicity, law invariance, cash subadditivity, quasi-convex, and quasi-concave in mixtures but not cash additive or positively homogeneous. 
For more details basic properties, we refer to \cite{BHWW25}.
}
\end{example}

Both $\VaR_\Lambda$ and $\ES_\Lambda$ can be written in the alternative form
\begin{equation*}
\rho_\Lambda(X) = \inf_{x \in \R} \left\{\rho_{\Lambda(x)}(X) \vee x\right\} 
\end{equation*}

The success of the $\Lambda$-framework for VaR and ES motivates the study of a Lambda version of the R{\'e}nyi entropic value-at-risk (EVaR), which we denote by $\EVaR_\Lambda^p$. 
In the next subsection we recall the definition of the (classical) R{\'e}nyi EVaR, and then in Sections \ref{sec3} and \ref{sec4} we introduce and analyze its $\Lambda$-counterpart.

\subsection{R{\'e}nyi entropic value-at-risk}

In this subsection we review the definition and basic properties of the R{\'e}nyi entropic value-at-risk (EVaR), which was introduced by \cite{PS20} and later extended to the adjusted version by \cite{ZWXH23}.
The construction of EVaR is based on the R{\'e}nyi entropy of order $q \in [1, \oo]$ for a probability measure $\Q$ with respect to $\P$, a concept that we first recall.

\begin{definition}[R{\'e}nyi entropy]
  \label{def-entropy1}
For $q \in [1, \oo]$, the R\'{e}nyi entropy of order $q$ for a probability measure $\Q$ with respect to $\P$ is defined as
\begin{equation}
  \label{eq-1103}
    H_q\(\frac{\d \Q}{\d \P}\):= \left\{\begin{array}{ll}
        \E \[\frac{\d \Q}{\d \P} \log \frac{\d \Q}{\d \P}\], &  {\rm for}\ q = 1\ {\rm and}\ \Q \ll \P, \\
        \log\left\| \frac{\d \Q}{\d \P} \right\|_{\oo}, & {\rm for}\ q = +\oo\ {\rm and}\ \Q \ll \P, \\[3pt]
        \displaystyle\frac{1}{q-1} \log \E \[\(\frac{\d \Q}{\d \P}\)^q\], & {\rm for}\ q \in (1, \oo)\ {\rm and}\ \Q \ll \P,  \\
\oo, & {\rm otherwise}, \end{array} \right.
\end{equation}
provided that the expectations are finite. Here we adopt the convention $0\log 0=0$.
\end{definition}

For a fixed $p \geq 1$ and its conjugate exponent $q$ satisfying $1/p + 1/q = 1$, then the entropic value-at-risk $\EVaR^p_{\alpha}: L^p \to \R$ of order $p$ at level $\alpha \in [0, 1)$ is defined as the following supremum over probability measures whose R\'{e}nyi entropy $H_q$ does not exceed $\log \frac{1}{1-\alpha}$:
\begin{equation*}
   \EVaR^p_{\alpha}(X):= \sup\left\{\E_\Q[X]: H_q\(\frac{\d \Q}{\d \P}\) \le \log\frac{1}{1-\alpha}\right \}.
\end{equation*}
For the limiting case $\alpha=1$, define $\EVaR^p_1(X)= \esssup (X)$ for any $p$.

As shown by \cite{PS20} \citep[see also][]{DPR10}, for $\alpha \in [0,1]$, the risk measure $\EVaR^p_{\alpha}$ with $p \geq 1$ admits an equivalent primal formulation that is often more convenient for computation and analysis:
\begin{equation}\label{EVaR}
\EVaR_\alpha^p(X) = \min_{t \in \overline{\R}} \left\{t + \(\frac{1}{1-\alpha}\)^{1/p} \(\E\[\(X - t\)_+^p\]\)^{1/p} \right\} ~~~ X \in L^p,
\end{equation}
where we use the conventions $0/0 = 0$ and $x/0 = \oo$ for $x > 0$. 
Associated with the minimisation problem \eqref{EVaR} are two auxiliary quantities that will play an important role later:
\begin{equation*}
   \VaR_{\alpha}^p(X) := \inf \underset{t \in \overline{\R}}{\arg\min} \left\{t + \(\frac{1}{1-\alpha}\)^{1/p} \(\E\[\(X - t\)_+^p\]\)^{1/p} \right\},
\end{equation*}
and
\begin{equation*}
     \VaR_{\alpha}^{p+}(X) := \sup \underset{t \in \overline{\R}}{\arg\min} \left\{t + \(\frac{1}{1-\alpha}\)^{1/p} \(\E\[\(X - t\)_+^p\]\)^{1/p} \right\}.
\end{equation*}
It follows \citep[see Lemma~A.1 of][]{ZWXH23} that the set of minimisers is exactly the interval
\begin{equation}\label{EVaR_minimizer}
\underset{t \in \overline{\R}}{\arg\min} \left\{t + \(\frac{1}{1-\alpha}\)^{1/p} \(\E\[\(X - t\)_+^p\]\)^{1/p} \right\} = \left\{\begin{array}{ll}
        \[\VaR_{\alpha}^p(X), \VaR_{\alpha}^{p+}(X)\], &  {\rm if}\ \alpha \in [0, 1), \\
        \VaR_1^p(X), & {\rm if}\ \alpha = 1,
\end{array} \right.
\end{equation}
with $\VaR_1^p(X) = \lim_{\alpha \to 1}  \VaR_\alpha^p(X) = \esssup(X)$. 
For more properties on functional $\VaR_{\alpha}^p$ or $\VaR_{\alpha}^{p+}$, we refer Section Appendice A.1 of \cite{ZWXH23}.
In particular, when $p=1$, one recovers the classical ES:
$$
   \VaR_{\alpha}^1(X) = \VaR_{\alpha}(X),\quad \VaR_{\alpha}^{1+}(X)=\VaR_{\alpha}^+(X)\ \ {\rm and}\ \  \EVaR^1_{\alpha}(X) = \ES_{\alpha}(X).
$$

Several fundamental properties of EVaR are worth highlighting.  
For each fixed $\alpha \in [0,1]$, the risk measure $\EVaR_\alpha^p$ is a law invariant, coherent risk measure on $L^p$; that is, it satisfies monotonicity, cash additivity, positive homogeneity, and convexity.
Moreover, it is continuous with respect to the $L^p$-norm \citep[][Proposition 2.5]{ZWXH23}, meaning that $\rho(X_n) \to \rho(X)$ for all $X, X_1, X_2, \cdots \in L^p$ and $X_n \overset{L^p}{\to} X$ as $n \to \oo$. 
Furthermore, $\EVaR_\alpha^p$ is consistent with the $(p+1)$-icx order, i.e., 
$$\EVaR_\alpha^p(X) \leq \EVaR_\alpha^p(Y) ~~~ {\rm whenever} ~~~ X \leq_{(p+1)\hbox{-} \rm icx} Y,$$ 
where $X \leq_{p\hbox{-} \rm icx} Y$ is defined by the condition
$$X \leq_{p\hbox{-} \rm icx} Y \iff \left\|(X - x)_+\right\|_{p-1} \leq \left\|(Y - x)_+\right\|_{p-1},\ \forall x \in \R.$$  
For more details on the $p$increasing convex ($p$-icx) order, we refer to \cite{SS07}. 
These properties make EVaR a natural and mathematically sound building block for its $\Lambda$-counterpart, which we introduce in Section \ref{sec3}.

\section{Lambda EVaR}\label{sec3}

We now introduce the main object of study in this paper: the Lambda extension of the R{\'e}nyi EVaR. 
The construction follows the general pattern described in Section \ref{sec22}, where a classical risk measure indexed by a confidence level is generalized by allowing the level to depend on a loss threshold via a decreasing function $\Lambda$.

\begin{definition}[$\Lambda$-EVaR]
For a decreasing function $\Lambda : \R \to [0, 1]$, the $\Lambda$-R{\'e}nyi entropic value-at-risk ({\rm $\Lambda$-EVaR}) is defined as the risk measure $\EVaR_\Lambda^p: L^p \to \overline{\R}$ with $p \geq 1$ given by
\begin{equation}\label{eq-2026-1111}
\EVaR_\Lambda^p(X) = \sup_{x \in \R} \left\{\EVaR_{\Lambda(x)}^p(X) \wedge x\right\} ~~~ X \in L^p.
\end{equation}
\end{definition}

The interpretation is analogous to that of $\Lambda$-VaR and $\Lambda$-ES: for each candidate loss level $x$, we compute the EVaR at the local confidence level $\Lambda(x)$, compare it with $x$ itself, and then take the supremum over all thresholds. 
The definition ensures that $\EVaR_\Lambda^p$ inherits the properties of the family $\{\EVaR_\alpha^p\}_{\alpha \in [0,1]}$ while introducing the flexibility to vary the confidence level with the magnitude of potential losses.

Before studying the detailed properties of $\Lambda$-EVaR, we first establish that the mapping is well defined on $L^p$ and clarify when it takes finite values. 
Consequently, the risk measure $\EVaR_\Lambda^p$ is always well-defined (though it may take the value $\oo$) on the space $L^p$.
The following proposition provides basic bounds and a finiteness criterion.

\begin{proposition}\label{prop-finite}
Let $\Lambda : \R \to [0, 1]$ be a decreasing function. 
The mapping $\EVaR_\Lambda^p: L^p \to \overline{\R}$ satisfies
\begin{equation*}
-\oo < \E[X] \leq \EVaR_\Lambda^p(X) \leq \EVaR_1^p(X) = \esssup(X) ~~~ X \in L^p.
\end{equation*}
In particular, $\EVaR_\Lambda^p(X)$ is finite on $L^p$ if and only if $\VaR_1(X) < \oo$ or $\Lambda$ is not constantly $1$.
Moreover, when $\Lambda$ is left-continuous and $\EVaR_\Lambda^p(X)$ is finite on $L^p$, the supremum in \eqref{eq-2026-1111} can be replaced by a maximum.
\end{proposition}
\begin{proof}
For any $x \in \R$, we have
\begin{equation*}
\E[X] \leq \EVaR_{\Lambda(x)}^p(X) \leq \EVaR_1^p(X) ~~~ X \in L^p.
\end{equation*}
Therefore, the relation $-\oo < \E[X] \leq \EVaR_\Lambda^p(X) \leq \EVaR_1^p(X)$ holds trivially.

If $\VaR_1(X) < \oo$, then $\EVaR_\Lambda^p(X) \leq \EVaR_1^p(X) = \esssup(X) = \VaR_1(X) < \oo$. 
Now assume that $\Lambda$ is not identically $1$. 
Then there exists $x_0 \in \R$, such that $0 \leq \Lambda(x_0) < 1$. 
Since $\Lambda$ is decreasing, we have $0 \leq \Lambda(x) \leq \Lambda(x_0) < 1$ for any $x \geq x_0$, and therefore
\begin{equation*}
\EVaR_{\Lambda(x)}^p(X) \leq \(\frac{1}{1-\Lambda(x)}\)^{1/p} \(\E\[|X|^p\]\)^{1/p} \leq \(\frac{1}{1-\Lambda(x_0)}\)^{1/p} \(\E\[|X|^p\]\)^{1/p} < \oo ~~~ x \geq x_0.
\end{equation*}
For $x < x_0$, we get $\EVaR_{\Lambda(x)}^p(X) \wedge x \leq x \leq x_0 < \oo$.
Thus $\EVaR_\Lambda^p(X)$ is finite on $L^p$ whenever $\Lambda$ is not constantly equal to $1$.
For the converse, suppose $\VaR_1(X) = \oo$ and $\Lambda \equiv 1$.
Then clearly $\EVaR_\Lambda^p(X) = \oo$, which would contradict the finiteness of $\EVaR_\Lambda^p$.

Finally, when $\Lambda$ is left-continuous and $\EVaR_\Lambda^p(X)$ is finite on $L^p$, we show that there exists some point $x^* \in \R$ such that $\EVaR_\Lambda^p(X) = \EVaR_{\Lambda(x^*)}^p(X) \wedge x^*$.
By the definition of $\EVaR_\Lambda^p$, there exists a sequence of $\{x_n\}_{n \in \N}$ such that  
$$\lim_{n \to \oo} \EVaR_{\Lambda(x_n)}^p(X) \wedge x_n = \EVaR_\Lambda^p(X).$$
If $\{x_n\}$ is a bounded sequence, then we can take its subsequence $\{x_{n_k}\}$ such that $x_{n_k} \uparrow x^*$ as $k \to \oo$.
Thus, $\lim_{k \to \oo} \Lambda(x_{n_k}) = \Lambda(x^*)$ by the left-continuous of $\Lambda$.
Consequencely,
\begin{equation*}
\EVaR_\Lambda^p(X) = \lim_{k \to \oo} \EVaR_{\Lambda(x_{n_k})}^p(X) \wedge x_{n_k} = \EVaR_{\Lambda(x^*)}^p(X) \wedge x^*.
\end{equation*}
If $\{x_n\}$ is an unbounded sequence. We first assume $\lim_{n \to \oo} x_n = -\oo$, then $\EVaR_\Lambda^p(X) = -\oo$.
This is impossible.
So, we next assume $\lim_{n \to \oo} x_n = \oo$. Then $\EVaR_\Lambda^p(X) = \lim_{n \to \oo} \EVaR_{\Lambda(x_n)}^p(X) \wedge x_n = \EVaR_{\Lambda(\oo)}^p(X)$.
Thus, we take $x^* = \EVaR_\Lambda^p(X)$ and $\EVaR_{\Lambda(x^*)}^p(X) \geq \EVaR_{\Lambda(\oo)}^p(X)$ $= \EVaR_\Lambda^p(X)$.
Consequencely,
\begin{equation*}
\EVaR_\Lambda^p(X) = \lim_{n \to \oo} \EVaR_{\Lambda(x_n)}^p(X) \wedge x_n = \EVaR_{\Lambda(\oo)}^p(X) = \EVaR_\Lambda^p(X) = \EVaR_{\Lambda(x^*)}^p(X) \wedge x^*.
\end{equation*}
This completes the proof.
\end{proof}

Proposition \ref{prop-finite} guarantees that $\Lambda$-EVaR is a proper risk measure on the appropriate domain. 
The next result provides a useful characterization of the points at which the supremum in \eqref{eq-2026-1111} is attained. 
This characterization plays a central role in many subsequent proofs and also leads to an alternative representation of $\Lambda$-EVaR.

\begin{proposition}\label{prop32}
Let $\Lambda : \R \to [0, 1]$ be a decreasing function.
Then for $X \in L^p$ and $x \in \R$,
\begin{equation}\label{eq-2026-1112}
\EVaR_{\Lambda(x^+)}^p(X) \leq x \leq \EVaR_{\Lambda(x^-)}^p(X)  \iff \EVaR_{\Lambda}^p(X) = x,
\end{equation}
and
\begin{equation}\label{eq-2026-1113}
\EVaR_{\Lambda(x)}^p(X) = x \Longrightarrow \EVaR_{\Lambda}^p(X) = x.
\end{equation}
If $\Lambda$ is continuous, then \eqref{eq-2026-1113} becomes an equivalence.
\end{proposition}
\begin{proof}
If $\Lambda \equiv 1$, \eqref{eq-2026-1112} and \eqref{eq-2026-1113} hold trivially.
Next, we consider that $\Lambda$ is not constantly $1$.
Let's first prove the equivalence stated in \eqref{eq-2026-1112}. 

($\Longrightarrow$) Assume $\EVaR_{\Lambda(x^+)}^p(X) \leq x \leq \EVaR_{\Lambda(x^-)}^p(X)$.
Using the definition of $\EVaR_\Lambda^p$, we obtain
\begin{equation*}
\EVaR_\Lambda^p(X) = \sup_{y \in \R} \left\{\EVaR_{\Lambda(y)}^p(X) \wedge y\right\} \geq \EVaR_{\Lambda(x^-)}^p(X) \wedge x = x.
\end{equation*}
To obtain the reverse inequality, we consider three cases.
\begin{itemize}
\item If $y > x$, then $\EVaR_{\Lambda(y)}^p(X) \leq \EVaR_{\Lambda(x^+)}^p(X) \leq x$, so,
\begin{equation*}
\EVaR_{\Lambda(y)}^p(X) \wedge y = \EVaR_{\Lambda(y)}^p(X) \leq x.
\end{equation*}

\item If $y < x$, then $\EVaR_{\Lambda(y)}^p(X) \geq \EVaR_{\Lambda(x^-)}^p(X) \geq x$, hence,
\begin{equation*}
\EVaR_{\Lambda(y)}^p(X) \wedge y = y < x.
\end{equation*}

\item If $y = x$, then trivially $\EVaR_{\Lambda(y)}^p(X) \wedge y = \EVaR_{\Lambda(x)}^p(X) \wedge x \leq x$.
\end{itemize}
Thus, $\EVaR_{\Lambda(y)}^p(X) \wedge y \leq x$ for every $y$, which yields $\EVaR_\Lambda^p(X) \leq x$.
Consequently, $\EVaR_\Lambda^p(X)=x$.

($\Longleftarrow$) Now suppose $\EVaR_\Lambda^p(X) = x$.
If we had $x < \EVaR_{\Lambda(x^+)}^p(X)$, then there would exist $y > x$ such that $\EVaR_{\Lambda(y)}^p(X) > x$ and $\EVaR_{\Lambda(y)}^p(X) \wedge y > x$.
Thus,
\begin{equation*}
\EVaR_{\Lambda}^p(X) \geq \EVaR_{\Lambda(y)}^p(X) \wedge y > x
\end{equation*}
contradicting $\EVaR_\Lambda^p(X) = x$.
Therefore, $x \geq \EVaR_{\Lambda(x^+)}^p(X)$. 
A symmetric argument (using values $y < x$) gives $x \leq \EVaR_{\Lambda(x^-)}^p(X)$.
This completes the proof of \eqref{eq-2026-1112}.

The implication \eqref{eq-2026-1113} follows directly from \eqref{eq-2026-1112}, because
$\EVaR_{\Lambda(x)}^p(X) = x$ implies $\EVaR_{\Lambda(x^+)}^p(X) \leq x \leq \EVaR_{\Lambda(x^-)}^p(X)$ when one notes that $\Lambda(x^+) \leq \Lambda(x) \leq \Lambda(x^-)$.

Finally, if $\Lambda$ is continuous, then $\EVaR_{\Lambda(x^+)}^p(X) = \EVaR_{\Lambda(x)}^p(X) = \EVaR_{\Lambda(x^-)}^p(X)$ by $\Lambda(x^+) = \Lambda(x) = \Lambda(x^-)$.
In this case \eqref{eq-2026-1112} reduces to $\EVaR_{\Lambda(x)}^p(X) = x \Longleftrightarrow \EVaR_{\Lambda}^p(X) = x$, which shows that \eqref{eq-2026-1113} is indeed an equivalence.
\end{proof}

The equivalence \eqref{eq-2026-1112} states that $x$ is the value of $\Lambda$-EVaR exactly when it lies between the left- and right-hand limits of the ordinary EVaR with confidence levels $\Lambda(x^-)$ and $\Lambda(x^+)$. 
This sandwich condition is often easier to verify than directly computing the supremum in \eqref{eq-2026-1111}.

An immediate consequence of Proposition \ref{prop32} is a dual representation of $\Lambda$-EVaR as an infimum rather than a supremum. 
This alternative form is particularly convenient for establishing quasi-concavity and for concave in mixtures arguments.

\begin{proposition}\label{prop-inf}
Let $\Lambda : \R \to [0, 1]$ be a decreasing function.
Then for $X \in L^p$ and $x \in \R$,
\begin{eqnarray}\label{eq-2026-1114}
\EVaR_\Lambda^p(X) = \inf_{x \in \R} \left\{\EVaR_{\Lambda(x)}^p(X) \vee x\right\}.
\end{eqnarray}
Moreover, the infimum in \eqref{eq-2026-1114} can be replaced by a minimum when $\Lambda$ is right-continuous and $\EVaR_\Lambda^p(X)$ is finite on $L^p$.
\end{proposition}
\begin{proof}
If $\Lambda$ is not constantly $1$.
Let $\EVaR_\Lambda^p(X) = x$.
By \eqref{eq-2026-1112}, we obtain $\EVaR_{\Lambda(x^+)}^p(X) \leq x \leq \EVaR_{\Lambda(x^-)}^p(X)$.
Consequently,
\begin{eqnarray*}
\inf_{y \in \R} \left\{\EVaR_{\Lambda(y)}^p(X) \vee y\right\} \leq \EVaR_{\Lambda(x^+)}^p(X) \vee x = x.
\end{eqnarray*}
To obtain the reverse inequality, it suffices to show that $\EVaR_{\Lambda(y)}^p(X) \vee y \geq x$ for all $y \in \R$ whenever $\EVaR_\Lambda^p(X) = x$. 
We consider two cases.
\begin{itemize}
\item \textbf{Case $y \geq x$:} Since $y \geq x$, we have
\begin{equation*}
\EVaR_{\Lambda(y)}^p(X) \vee y \geq y \geq x.
\end{equation*}

\item \textbf{Case $y < x$:} Because $y < x$, the monotonicity of $\Lambda$ gives $\Lambda(y) \geq \Lambda(x^-)$.
Hence $\EVaR_{\Lambda(y)}^p(X) \geq \EVaR_{\Lambda(x^-)}^p(X) \geq x$, and therefore
\begin{equation*}
\EVaR_{\Lambda(y)}^p(X) \vee y \geq \EVaR_{\Lambda(y)}^p(X) \geq x.
\end{equation*}
\end{itemize}
Thus $\EVaR_{\Lambda(y)}^p(X) \vee y \geq x$ for all $y$, which implies
\begin{eqnarray*}
\inf_{y \in \R} \left\{\EVaR_{\Lambda(y)}^p(X) \vee y\right\} \geq x.
\end{eqnarray*}
Therefore, we conclude that $\inf_{y \in \R} \left\{\EVaR_{\Lambda(y)}^p(X) \vee y\right\} = \EVaR_\Lambda^p(X)$.

If $\Lambda \equiv 1$, then $\EVaR_\Lambda^p(X) = \esssup(X)$ and
\begin{equation*}
\EVaR_\Lambda^p(X) =\esssup(X) = \inf_{x \in \R} \left\{\esssup(X) \vee x\right\} = \inf_{x \in \R} \left\{\EVaR_{\Lambda(x)}^p(X) \vee x\right\}.
\end{equation*}
The proof that the infimum in \eqref{eq-2026-1114} can be replaced by a minimum follows an argument similar to that of Proposition \ref{prop-finite} and is therefore omitted.
Thus, we complete the proof.
\end{proof}

The following proposition collects its most important properties, showing that it inherits many desirable characteristics from the underlying EVaR family while also exhibiting the flexibility typical of Lambda-type risk measures.

\begin{proposition}
For any decreasing functions $\Lambda, \Lambda' : \R \to [0, 1]$, the risk measure $\EVaR_\Lambda^p$ satisfies the following properties: 
\begin{enumerate}[{\rm (i)}]
\item $\EVaR_\Lambda^p \geq \EVaR_{\Lambda'}^p$ when $\Lambda \geq \Lambda'$;
\item $\EVaR_\Lambda^p$ is monotone;
\item $\EVaR_\Lambda^p \geq \ES_\Lambda \geq \VaR_\Lambda$;
\item $\EVaR_\Lambda^p$ is quasi-convex;
\item $\EVaR_\Lambda^p$ is normalized;
\item $\EVaR_\Lambda^p$ is cash subadditive;
\item $\EVaR_\Lambda^p$ is $(p+1)$-consistent;
\item $\EVaR_\Lambda^p$ is quasi-concave in mixtures;
\item $\EVaR_\Lambda^p$ is $L^p$-continuous when $\Lambda$ takes values in $[0, 1)$.
\end{enumerate}
\end{proposition}
\begin{proof}
(i) and (ii) are straightforward because $\EVaR_\alpha^p(X)$ is monotone (increasing) in both $\alpha \in [0, 1]$ and $X$, and the supremum of monotone transformations of $\EVaR_\alpha^p(X)$ is also monotone.
(iii) follows from Proposition 2.2 of \cite{ZWXH23} and Proposition 2 of \cite{BHWW25}.
(iv) holds because an increasing transform of a quasi-convex function is quasi-convex, the supremum of a family of quasi-convex functions is quasi-convex, and $\EVaR_\alpha^p$ is quasi-convex for each fixed $\alpha \in [0, 1]$.
(v) is a direct consequence of \eqref{eq-2026-1113}.

To see (vi), let $X \in L^p$ and $c \in \R_+$. Then
\begin{eqnarray*}
\EVaR_\Lambda^p(X + c) &=& \sup_{x \in \R} \left\{\EVaR_{\Lambda(x)}^p(X + c) \wedge x\right\} = \sup_{x \in \R} \left\{\EVaR_{\Lambda(x)}^p(X) + c \wedge x\right\}  \\
&\leq& \sup_{x \in \R} \left\{\EVaR_{\Lambda(x)}^p(X) \wedge x\right\} + c = \EVaR_\Lambda^p(X)  + c.
\end{eqnarray*}

To see (vii), note that $X \leq_{\rm p\hbox{-} icx} Y$ implies $\EVaR_\alpha^p(X) \leq \EVaR_\alpha^p(Y)$ for any $\alpha \in [0,1]$. 
Consequently, $\EVaR_\Lambda^p(X) \leq \EVaR_\Lambda^p(Y)$.

To prove (viii), we note that quasi-concavity is preserved under increasing transformations, and the infimum of a family of quasi-concave in mixtures functions is again quasi-concave in mixtures.
In view of representation \eqref{eq-2026-1114}, it therefore suffices to show that $\EVaR_\alpha^p \vee x$ is quasi-concave in mixtures for each $\alpha \in [0, 1]$.
It is only to show $\EVaR_\alpha^p$ is quasi-concave.
Take independent event $A$ with $X$ and $Y$ such that $\P(A) = \lm$ and let $Z = \I_A X + \I_{A^c} Y$.
Then by the concavity of $x \mapsto x^{1/p}$ for $p \geq 1$ and \eqref{EVaR},
\begin{eqnarray*}
\EVaR_\alpha^p(Z) &=& \inf_{t \in \R} \left\{t + \(\frac{1}{1-\alpha}\)^{1/p} \(\E\[\(Z - t\)_+^p\]\)^{1/p} \right\}  \\
&=& \inf_{t \in \R} \left\{t + \(\frac{1}{1-\alpha}\)^{1/p} \(\lm \E\[\(X - t\)_+^p\] + (1-\lm) \E\[\(Y - t\)_+^p\]\)^{1/p} \right\}  \\
&\geq& \inf_{t \in \R} \left\{t + \(\frac{1}{1-\alpha}\)^{1/p} \(\lm \(\E\[\(X - t\)_+^p\]\)^{1/p} + (1-\lm) \(\E\[\(Y - t\)_+^p\]\)^{1/p}\) \right\}  \\
&\geq& \lm \EVaR_\alpha^p(X) + (1-\lm) \EVaR_\alpha^p(Y).
\end{eqnarray*}

To prove item (ix), first note that $\EVaR_\alpha^p$ is $L^p$-continuous \citep[e.g.,][Proposition 2.5]{ZWXH23} for each $\alpha \in [0, 1)$. 
Take any random variable $X$ and any sequence $(X_n)_{n \in \N}$ in $L^p$ such that $X_n \to X$ in $L^p$ as $n \to \oo$. 
Let $f_n : x \to \EVaR_{\Lambda(x)}^p(X_n) - x$ and $f : x \to \EVaR_{\Lambda(x)}^p(X) - x$. 
By \eqref{eq-2026-1111} and \eqref{eq-2026-1112}, for any $y, z$ with $y < \EVaR_{\Lambda(x)}^p(X) < z$, we have $f(y) > 0 > f(z)$. 
Therefore, because $f_n \to f$ pointwise, we have $f_n(y) > 0 > f_n(z)$ for $n$ large enough. 
This implies $y \leq \EVaR_\Lambda^p(X_n) \leq z$ via \eqref{eq-2026-1112}. 
Since $y, z$ are arbitrarily close to $\EVaR_\Lambda^p(X)$, we know $\EVaR_\Lambda^p(X_n) \to \EVaR_\Lambda^p(X)$.
\end{proof}

A natural question is under which conditions $\Lambda$-EVaR coincides with a classical (non-Lambda) risk measure. 
The following proposition gives a complete answer: convexity, concavity in mixtures, and cash additivity are all equivalent to $\Lambda$ being constant. 
This result parallels those known for $\Lambda$-VaR \citep[][Proposition 2]{XH25} and $\Lambda$-ES \citep[][Proposition 3]{BHWW25}, confirming that the Lambda framework genuinely extends the classical setting only when $\Lambda$ is non-constant.

\begin{proposition}\label{prop35}
For any decreasing function $\Lambda: \R \to [0, 1]$, the following are equivalent.
\begin{enumerate}[{\rm (i)}]
\item The risk measure $\EVaR_\Lambda^p$ is convex.
\item The risk measure $\EVaR_\Lambda^p$ is concave in mixtures.
\item The risk measure $\EVaR_\Lambda^p$ is cash additivity.
\item The function $\Lambda$ is constant on $\R$.
\end{enumerate}
\end{proposition}
\begin{proof}
It suffices to prove (i) $\Longrightarrow$ (iv), (ii) $\Longrightarrow$ (iv) and (iii) $\Longrightarrow$ (iv), because the implications (iv) $\Longrightarrow$ (i), (iv) $\Longrightarrow$ (ii) and (iv) $\Longrightarrow$ (iii) follow directly from the fact that $\EVaR_\alpha^p$ is both convex, concave in mixtures and cash additivity.
Next, we prove (i) $\Longrightarrow$ (iv) by contradiction.
Suppose that $\Lambda$ is decreasing and not constant on $\R$. 
Then there exist $x > y$ such that $\Lambda(x^-) < \Lambda((x+y)/2) \leq \Lambda(y)$. 
Define a random variable $X$ by $\P(X = y) = \Lambda((x+y)/2) = 1 - \P(X = x)$, and let $Y$ be the constant random variable with $\P(Y = y) = 1$.
Set $Z = (X + Y)/2$. 
For any $t < y$, we have
\begin{equation*}
t + \(\frac{1}{1-\Lambda((x+y)/2)}\)^{1/p} \(\Lambda\(\frac{x+y}{2}\) (y-t)^p + \(1-\Lambda\(\frac{x+y}{2}\)\) \(\frac{x+y}{2}-t\)^p\)^{1/p} \geq \frac{x+y}{2}.
\end{equation*}
Consequently, $\EVaR_{\Lambda((x+y)/2)}^p(Z) = (x+y)/2$ and $\EVaR_\Lambda^p(Z) = (x+y)/2$ by \eqref{eq-2026-1113}. 
Moreover, $\P(Y = y) = 1$ implies $\EVaR_\Lambda^p(Y) = y$.
To reach a contradiction, we must show that $\EVaR_\Lambda^p(X) < x$.
By \eqref{eq-2026-1112}, it is enough to prove $\EVaR_{\Lambda(x^-)}^p(X) < x$.
To see this, take any $u \in (x, y)$. Then
\begin{equation*}
\EVaR_{\Lambda(x^-)}^p(X) \leq u + \(\frac{1 - \Lambda\((x+y)/2\)}{1-\Lambda(x^-)}\)^{1/p} \(x - u\)  < u + \(x - u\)  = x,
\end{equation*} 
as required.

Third, we prove (ii) $\Longrightarrow$ (iv) by contradiction. 
Suppose that $\Lambda$ is not constant on $\R$. Since $\Lambda$ is bounded, it cannot be concave.
Consequently, there exist distinct points $x, y, z \in \R$ and $\th \in (0, 1)$ such that 
$$z = \th x + (1 - \th)y ~~~ {\rm and} ~~~ \Lambda(z) < \th \Lambda(x) + (1 - \th)\Lambda(y).$$ 
By the continuity of linear functions, we can find $\gamma \in (0, 1)$ arbitrarily close to $\th$ such that 
$$z < \gamma x + (1 - \gamma)y ~~~ {\rm and} ~~~ \Lambda(z) < \gamma \Lambda(x) + (1 - \gamma)\Lambda(y).$$ 
Take independent events $A, B, C \in {\cal F}$ with $\P(A) = 1 - \Lambda(x), \P(B) = 1 - \Lambda(y)$, and $\P(C) = \gamma$. 
Choose a constant $K > \max\{-x, -y\}$ (to be fixed later) and define
\begin{equation*}
X = x \I_A - K \I_{A^c}, ~~~ Y = y \I_B - K \I_{B^c}, ~~~ {\rm and} ~~~ Z = \I_C X + \I_{C^c} Y.
\end{equation*}
For any $t < -K$, we have
\begin{equation*}
t + \(\frac{1}{1-\Lambda(x)}\)^{1/p} \(\Lambda(x) (-K-t)^p + \(1-\Lambda(x)\) \(x-t\)^p\)^{1/p} \geq x.
\end{equation*}
Therefore, $\EVaR_{\Lambda(x)}^p(X) = x$ and $\EVaR_\Lambda^p(X) = x$ by \eqref{eq-2026-1113}. 
Similarly, $\EVaR_\Lambda^p(Y) = y$.
To reach a contradiction, we show that for sufficiently large $K$ we have $\EVaR_\Lambda^p(Z) \leq z$.
Observe that
\begin{eqnarray*}
\EVaR_{\Lambda(z)}^p(Z) &\leq& -K + \(\frac{1}{1-\Lambda(z)}\)^{1/p} \(\gamma (1-\Lambda(x)) (x + K)^p + (1 - \gamma) \(1-\Lambda(y)\) (y + K)^p\)^{1/p}    \\
&\leq& -K + \(\frac{1}{1-\Lambda(z)}\)^{1/p} \(\gamma (1-\Lambda(x)) + (1 - \gamma) (1-\Lambda(y))\)^{1/p} (x \vee y + K),
\end{eqnarray*}
which tends to $-\oo$ as $K \to \oo$ because $\Lambda(z) < \gamma \Lambda(x) + (1 - \gamma)\Lambda(y)$.
In particular, for large enough $K$, we have $\EVaR_{\Lambda(z)}^p(Z) \leq z$.
Using \eqref{eq-2026-1114}, we conclude that $\EVaR_\Lambda^p(Z) \leq z$.

Finally, we prove (iii) $\Longrightarrow$ (iv) by contradiction. 
Suppose that $\Lambda$ is decreasing and not constant on $\R$. 
Then there exist $x > y$ such that $\Lambda(x^-) < \Lambda((x+y)/2) \leq \Lambda(y)$. 
Define a random variable $X$ by $\P(X = y) = \Lambda((x+y)/2) = 1 - \P(X = (x+y)/2)$.
Then $\EVaR_{\Lambda((x+y)/2)}^p(X) = (x+y)/2$ and $\EVaR_\Lambda^p(X) = (x+y)/2$ by \eqref{eq-2026-1113}.
Let $m = (x - y)/2 > 0$. 
Note that $\EVaR_{\Lambda(x^-)}^p(X + m) < \EVaR_{\Lambda((x+y)/2)}^p(X+m) = \EVaR_{\Lambda((x+y)/2)}^p(X) + m = x$.
Then $\EVaR_\Lambda^p(X + m) < x = (x+y)/2 + m$ by \eqref{eq-2026-1112}.
Hence, this contradicts the cash additivity of $\EVaR_\Lambda^p$.
\end{proof}

Another classical property often examined in risk measure theory is positive homogeneity. 
For $\Lambda$-EVaR, positive homogeneity imposes a very restrictive structure on $\Lambda$: 
it must be piecewise constant on the positive and negative half-lines.

\begin{proposition}
Let $\Lambda: \R \to [0, 1]$ be decreasing. 
If $\EVaR_\Lambda^p$ is positively homogeneous for all $X \in L^p$, then $\Lambda$ is constant on intervals $(0, \oo)$ and $(-\oo, 0)$, respectively, that is, there exist $1 \geq \alpha_1 \geq \alpha_2 \geq \alpha_3 \geq 0$ such that
\begin{equation}\label{piece}
\Lambda(x) = \alpha_1 \I_{(-\oo, 0)}(x) + \alpha_2 \I_{\{0\}}(x) + \alpha_3 \I_{(0, \oo)}(x).
\end{equation}
\end{proposition}
\begin{proof}
($\Longleftarrow$): Suppose $\Lambda$ has the form given in \eqref{piece}. Then by \eqref{eq-2026-1111}
\begin{eqnarray*}
\EVaR_\Lambda^p(X) &=& \max\left\{\sup_{x < 0} \left\{\EVaR_{\alpha_1}^p(X) \wedge x\right\}, \EVaR_{\alpha_2}^p(X) \wedge 0, \sup_{x > 0} \left\{\EVaR_{\alpha_3}^p(X) \wedge x\right\}  \right\}  \\
&=& \max\left\{\EVaR_{\alpha_1}^p(X) \wedge 0, \EVaR_{\alpha_2}^p(X) \wedge 0, \EVaR_{\alpha_3}^p(X) \right\}.
\end{eqnarray*}
From this representation it is clear that $\EVaR_\Lambda^p(\lm X) = \lm \EVaR_\Lambda^p(X)$ for $\lm \in (0, \oo)$.

($\Longrightarrow$): Assume that $\EVaR_\Lambda^p$ is positively homogeneous for all $X \in L^p$. 
Then for every $\lm \in (0, \oo)$
\begin{equation*}
\EVaR_\Lambda^p(\lm X) = \sup_{x \in \R} \left\{\lm \EVaR_{\Lambda(x)}^p(X) \wedge x\right\} = \lm \sup_{x \in \R} \left\{\EVaR_{\Lambda(x)}^p(X) \wedge x\right\} = \lm \EVaR_\Lambda^p(X).
\end{equation*}
Using the identity
$$\sup_{x \in \R} \left\{\lm \EVaR_{\Lambda(x)}^p(X) \wedge x\right\} = \lm \sup_{x \in \R} \left\{\EVaR_{\Lambda(\lm x)}^p(X) \wedge x\right\},$$
we obtain
\begin{equation}\label{eq-2026-11116}
\sup_{x \in \R} \left\{\EVaR_{\Lambda(\lm x)}^p(X) \wedge x\right\} = \sup_{x \in \R} \left\{\EVaR_{\Lambda(x)}^p(X) \wedge x\right\}.
\end{equation}
Now take $x > 0$ and let $\lm_1, \lm_2 > 0$. 
Setting $\lm = \lm_1 / \lm_2$ in \eqref{eq-2026-11116} yields
\begin{equation*}
\sup_{x \in \R} \left\{\EVaR_{\Lambda(\lm_1 x)}^p(X) \wedge x\right\} = \sup_{x \in \R} \left\{\EVaR_{\Lambda(\lm_2 x)}^p(X) \wedge x\right\}.
\end{equation*}
Consequently, $\Lambda(\lm_1 x) = \Lambda(\lm_2 x)$ for all $x > 0$.
Hence $\Lambda$ is constant on the positive half-line; denote this constant by $\alpha_3$.
A symmetric argument shows that $\Lambda$ is also constant on the negative half-ine; denote that constant by $\alpha_1$.
Finally, set $\alpha_2 := \Lambda(0)$. 
Because $\Lambda$ is decreasing, we have $1 \geq \alpha_1 \geq \alpha_2 \geq \alpha_3 \geq 0$.
Thus, $\Lambda$ necessarily has the piecewise-constant form given in \eqref{piece}, which completes the proof.
\end{proof}

The results of this section demonstrate that $\Lambda$-EVaR shares the core features of existing Lambda risk measures while being based on the more general R\'{e}nyi entropic construction.

\section{Representations and worst-case analysis}\label{sec4}

This section is devoted to alternative representations of $\Lambda$-EVaR and its behaviour under model uncertainty. 
We first present a dual representation that reveals the underlying extremal probability measures, then derive an extended Rockafellar-Uryasev formula that facilitates numerical computation. 
Finally, we examine worst-case values of $\Lambda$-EVaR over two classical uncertainty sets: those described by a Wasserstein distance and those specified by given mean and variance.

\subsection{Dual representation}

Dual representations are fundamental in convex risk measure theory because they express the risk measure as a supremum over a set of scenarios (probability measures). 
For $\Lambda$-EVaR, the dual representation involves probability measures whose R\'{e}nyi entropy is controlled by the function $\Lambda$. 
Let $\mathcal{M}_1(\Omega)$ denote the set of all probability measures on $(\Omega,\mathcal{F})$.

\begin{theorem}\label{thm41}
For any decreasing function $\Lambda: \R \to [0, 1]$, the risk measure $\EVaR_\Lambda^p$ adopts the following representation:
\begin{equation}\label{eq-2026-41}
\EVaR_\Lambda^p(X) = \sup_{\Q \in {\cal M}^q_\P} R\(\E_\Q[X], \Q\) ~~~ X \in L^p,
\end{equation}
where ${\cal M}^q_\P = \left\{\Q \in {\cal M}_1(\Omega): \Q \ll \P,\ \d \Q / \d \P \in L^q \right\}$ and for $(t, \Q) \in \R \times {\cal M}^q_\P$,
\begin{equation}\label{eq-2026-42}
R\(t, \Q\) = \sup_{x \in \R} \left\{t \wedge x: \Lambda(x) \geq 1 - \exp\(-H_q\(\frac{\d \Q}{\d \P}\)\) \right\},
\end{equation}
where $H_q$ is given in \eqref{eq-1103}. Moreover, the following statements hold.
\begin{enumerate}[{\rm (i)}]
\item The supremum in \eqref{eq-2026-42} can be changed to a maximum if $\Lambda$ is left-continuous.
\item $(t, \Q) \to R(t, \Q)$ is upper semicontinuous, quasi-concave, and increasing in $t$.
\item $\inf_{t \in \R} R(t, \Q) = \inf_{t \in \R} R(t, \Q')$ for all $\Q, \Q' \in {\cal M}^q_\P$.
\item $R(t_1, \Q) - R(t_2, \Q) \leq t_1 - t_2$ for all $t_1 \geq t_2$ and $\Q \in {\cal M}^q_\P$.
\end{enumerate}
\end{theorem}
\begin{proof}
Define
\begin{equation*}
{\cal M}^q_{\Lambda(x)} = \left\{\Q \in {\cal M}_1(\Omega): \frac{\d \Q}{\d \P} \in L^q, H_q\(\frac{\d \Q}{\d \P}\) \leq \log \frac{1}{1-\Lambda(x)} \right\} ~~~ x \in \R.
\end{equation*}
For any $X \in L^p$ and $x \in \R$, by Theorem 12 of \cite{PS20}, we have
\begin{eqnarray*}
\EVaR_\Lambda^p(X) &=& \sup_{x \in \R} \left\{\EVaR_{\Lambda(x)}^p(X) \wedge x\right\}   \\
&=& \sup_{x \in \R} \left\{\(\sup_{\Q \in {\cal M}^q_{\Lambda(x)}} \E_\Q[X]\) \wedge x\right\}   \\
&=& \sup_{x \in \R} \sup_{\Q \in {\cal M}^q_{\Lambda(x)}} \left\{\E_\Q[X] \wedge x\right\}   \\
&=& \sup_{\Q \in {\cal M}^q_\P} \sup_{x \in \R} \left\{\E_\Q[X] \wedge x: \Lambda(x) \geq 1 - \exp\(-H_q\(\frac{\d \Q}{\d \P}\)\)\right\}.
\end{eqnarray*}
Define
\begin{equation*}
A = \left\{x \in \R: \Lambda(x) \geq 1 - \exp\(-H_q\(\frac{\d \Q}{\d \P}\)\) \right\},
\end{equation*}
and $a := \sup A \in \R \cup \{-\oo, \oo\}$.
By the decreasing and left-continuous of $\Lambda$, we get $A = (-\oo, a]$.
Thus, $R(t, \Q) = \sup_{x \leq a} t \wedge x = \max_{x \leq a} t \wedge x$. This implies (i).
The proofs of (ii)-(iv) follow arguments analogous to those in Theorem 3 of \cite{BHWW25}.
\end{proof}

The function $R(t, \Q)$ can be interpreted as the maximal threshold $x$ for which the confidence level $\Lambda(x)$ is at least $1-\exp(-H_q(\d \Q/\d \P))$. 
Theorem \ref{thm41} thus expresses $\Lambda$-EVaR as a two-layer supremum: over probability measures $\Q$ with finite R\'{e}nyi entropy, and over thresholds $x$ compatible with that entropy. 
Properties (i)-(iv) reflect the regularity of this representatio.

\subsection{Extended Rockafellar-Uryasev formula}

The Rockafellar-Uryasev formula provides a powerful tool for computing and optimizing classical risk measures ES. 
For $\Lambda$-EVaR we obtain an analogous formula that involves an additional minimization over the threshold variable $x$.

Define $T_\Lambda: \overline{\R} \times \R \times L^p \to \overline{\R}$ by
\begin{equation}\label{T_Lambda}
T_\Lambda: (t, x, X) \mapsto \(t + \(\frac{1}{1-\Lambda(x)}\)^{1/p} \left\| (X-t)_+ \right\|_p\) \vee x
\end{equation}

\begin{theorem}
Let $\Lambda : \R \to [0, 1]$ be a right-continuous decreasing function and $\EVaR_\Lambda^p(X)$ is finite on $L^p$ and $T_\Lambda$ be given in \eqref{T_Lambda}. Then
\begin{equation}\label{EVaR_US}
\EVaR_\Lambda^p(X) = \min_{(t, x) \in \overline{\R} \times \R} T_\Lambda(t, x, X) = \min_{(t, x) \in \overline{\R} \times \R} \left\{ \(t + \(\frac{1}{1-\Lambda(x)}\)^{1/p} \left\| (X-t)_+ \right\|_p\) \vee x\right\} ~~~ X \in L^p,
\end{equation}
where the minima are obtained at $x^* = \EVaR_\Lambda^p(X)$ and
\begin{equation*}
t^* \left\{\begin{array}{ll}
        \in \[\VaR_{\Lambda(x^*)}^p(X), \VaR_{\Lambda(x^*)}^{p+}(X)\], &  {\rm if}\ \Lambda(x^*) \in [0, 1), \\
        = \VaR_1^{p+}(X), & {\rm if}\ \Lambda(x^*) = 1. 
\end{array} \right.
\end{equation*}
Moreover,
\begin{enumerate}[{\rm (i)}]
\item $T_\Lambda(t, x, X)$ is jointly convex in $(t, X) \in \overline{\R} \times L^p$ for all $x \in \R$;
\item $T_\Lambda(t, x, X)$ is convex in $x \in \R$ for all $(t, X) \in \overline{\R} \times L^p$ if and only if the function $x \mapsto \(1/(1 - \Lambda(x))\)^{1/p}$ is convex;
\item the following statements are equivalent:
\begin{enumerate}[{\rm (a)}]
\item $T_\Lambda(t, x, X)$  is jointly convex in $(t, x) \in \overline{\R} \times \R$ for all $X \in L^p$;
\item $T_\Lambda(t, x, X)$  is jointly quasi-convex in $(t, x) \in \overline{\R} \times \R$ for all $X \in L^p$;
\item $T_\Lambda(t, x, X)$  is jointly convex in $(t, x, X) \in \overline{\R} \times \R \times L^p$;
\item $T_\Lambda(t, x, X)$  is jointly quasi-convex in $(t, x, X) \in \overline{\R} \times \R \times L^p$;
\item $\Lambda$ is constant on $\R$.
\end{enumerate}
\end{enumerate}
\end{theorem}
\begin{proof}
For any $X \in L^p$, we have by Proposition \ref{prop-inf} and \eqref{EVaR} that
\begin{eqnarray*}
\EVaR_\Lambda^p(X) &=& \min_{x \in \R} \left\{\EVaR_{\Lambda(x)}^p(X) \vee x\right\}   \\
&=& \min_{x \in \R} \left\{\min_{t \in \overline{\R}} \(t + \(\frac{1}{1-\Lambda(x)}\)^{1/p} \left\| (X-t)_+ \right\|_p\) \vee x \right\}   \\
&=& \min_{x \in \R} \min_{t \in \overline{\R}} \left\{\(t + \(\frac{1}{1-\Lambda(x)}\)^{1/p} \left\| (X-t)_+ \right\|_p\) \vee x \right\}.
\end{eqnarray*}
Thus, \eqref{EVaR_US} holds.
By Proposition \ref{prop-inf}, we know that there exists $x^* \in \R$ such that 
$$\EVaR_\Lambda^p(X) = \EVaR_{\Lambda(x^*)}^p(X) \vee x^*.$$
We consider two cases.
\begin{itemize}
\item \textbf{Case $\EVaR_{\Lambda(x^*)}^p(X) < x^*$:} Since $\EVaR_{\Lambda(x^*)}^p(X) < x^*$, we have $x^* = \EVaR_\Lambda^p(X)$.

\item \textbf{Case $\EVaR_{\Lambda(x^*)}^p(X) \geq x^*$:} Then $\EVaR_\Lambda^p(X) = \EVaR_{\Lambda(x^*)}^p(X) =: S$ by $\EVaR_{\Lambda(x^*)}^p(X) \geq x^*$.
By \eqref{eq-2026-1112} and right-continuous of $\Lambda$, we have $\EVaR_{\Lambda(S)}^p(X) \leq S \leq \EVaR_{\Lambda(S^-)}^p(X)$.
Then
\begin{equation*}
\EVaR_{\Lambda(S)}^p(X) \vee S = S = \EVaR_{\Lambda(x^*)}^p(X) \vee x^*.
\end{equation*}
This implies that $S$ is also a minimizer.
Note that the intersection point between the graph (linearly interpolated) of the function $x \mapsto \EVaR_{\Lambda(x)}^p(X)$ and the graph of the identity is unique.
Therefore $x^* = S$.
\end{itemize}
This confirms that the minimizer is given by $x^* \mapsto \EVaR_{\Lambda(x^*)}^p(X)$, and the corresponding minimizer $t^*$ is obtained via \eqref{EVaR_minimizer}.

Property (i) follows directly from the Minkowski inequality, so we omit its proof.

(ii) If the function $x \mapsto \(1/(1 - \Lambda(x))\)^{1/p}$ is convex, then $T_\Lambda(t, x, X)$ is convex in $x \in \R$ for all $(t, X) \in \overline{\R} \times L^p$ by the definition of convexity.
Below, we show the ``only if" part.
Let $\tilde{x} = \inf\{x \in \R : \Lambda(x) = 0\}$. 
Right-continuity of $\Lambda$ yields that $\Lambda(\tilde{x}) = 0$. 
We first prove $x \mapsto \(1/(1 - \Lambda(x))\)^{1/p}$ is convex in $x \in (-\oo, \tilde{x})$. 
Suppose for contradiction that for some $x_0, y_0 \in (-\oo, \tilde{x})$
\begin{equation}\label{eq-cx1}
\(\frac{1}{1 - \Lambda\(\frac{x_0+y_0}{2}\)}\)^{1/p} > \frac{1}{2}\(\frac{1}{1 - \Lambda(x_0)}\)^{1/p} + \frac{1}{2} \(\frac{1}{1 - \Lambda(y_0)}\)^{1/p}.
\end{equation}
It is easy to show that
\begin{equation*}
\lim_{t \to -\oo} \(t + \(\frac{1}{1-\Lambda(x)}\)^{1/p} \left\| (X-t)_+ \right\|_p\) = \oo ~~~ \forall x < \tilde{x}.
\end{equation*}
Then, there exists $t_0 \in \R$ such that
\begin{equation*}
t_0 + \(\frac{1}{1-\Lambda(x_0)}\)^{1/p} \left\| (X-t_0)_+ \right\|_p \geq x_0, ~~~ t_0 + \(\frac{1}{1-\Lambda(y_0)}\)^{1/p} \left\| (X-t_0)_+ \right\|_p \geq y_0,
\end{equation*}
and
\begin{equation}\label{eq-cx2}
t_0 + \(\frac{1}{1-\Lambda((x_0+y_0)/2)}\)^{1/p} \left\| (X-t_0)_+ \right\|_p \geq \frac{x_0+y_0}{2}.
\end{equation}
It is clear that \eqref{eq-cx2} and \eqref{eq-cx2} together contradict the fact that \eqref{T_Lambda} is convex in $x$. 
Thus, $x \mapsto \(1/(1 - \Lambda(x))\)^{1/p}$ is convex in $x \in (-\oo, \tilde{x})$. 
The continuous of $x \mapsto \(1/(1 - \Lambda(x))\)^{1/p}$ at $x = \tilde{x}$ is similar to Theorem 4 of \cite{BHWW25}.
Therefore, $x \mapsto \(1/(1 - \Lambda(x))\)^{1/p}$ is convex in $x \in \R$. 

Next, we prove statement (iii). It is straightforward that (a) $\Longrightarrow$ (b) and (c) $\Longrightarrow$ (d).
Below, we show (e) $\Longrightarrow$ (a) and (e) $\Longrightarrow$ (c).
If $\Lambda$ is constant, say $\Lambda(x) \equiv = \alpha$, then $\(1/(1-\Lambda(x))\)^{1/p}$ is constant and $T_\Lambda(t, x, X) = \(t + \(1/(1-\alpha)\)^{1/p} \left\| (X-t)_+ \right\|_p\) \vee x$. 
The first term is convex in $(t,X)$ (by Minkowski's inequality), and the pointwise maximum with the linear function $x$ preserves convexity. 
Hence $T_\Lambda$ is jointly convex in $(t,x,X)$, which implies both (a) and (c).
Next, we show (b) $\Longrightarrow$ (e).
Suppose $T_\Lambda$ is jointly quasi-convex in $(t,x)$ for every $X \in L^p$ and $\Lambda$ is decreasing and non-constant on $\R$ by contradiction. 
Then there exists $y < x \leq t$ such that $\Lambda(y) \geq \Lambda((x+y)/2) > \Lambda(x)$. 
Choose a random variable $X$ with $\P(X = a) = 1 - \P(X = t) = \Lambda(x)$ with $a < t$.
By $\Lambda(x) < 1$, we have
\begin{equation*}
\(a + \(\frac{1}{1-\Lambda(x)}\)^{1/p} \left\| (X-a)_+ \right\|_p\) \vee x = \(t + \(\frac{1}{1-\Lambda(x)}\)^{1/p} \left\| (X-t)_+ \right\|_p\) \vee x = t,
\end{equation*}
whereas
\begin{eqnarray*}
& &\(\frac{a+t}{2} + \(\frac{1}{1-\Lambda((x+y)/2)}\)^{1/p} \left\| \(X-\frac{a+t}{2}\)_+ \right\|_p\) \vee \frac{x+y}{2} \\
& &\quad = \(\frac{a+t}{2} + \(\frac{1-\Lambda(x)}{1-\Lambda((x+y)/2)}\)^{1/p} \(t-\frac{a+t}{2}\)\) \vee \frac{x+y}{2} > t,
\end{eqnarray*}
This contradicts the joint quasi-convexity of \eqref{T_Lambda} in $(t, x) \in \overline{\R} \times \R$ for all $X \in L^p$, and thus $\Lambda$ is constant on $\R$.
Finally, note that joint quasi-convexity in $(t,x,X)$ in particular implies quasi-convexity in $(t,x)$ for each fixed $X$. As shown in (b) $\Longrightarrow$ (e), this forces $\Lambda$ to be constant.
\end{proof}

\subsection{Worst-case Lambda EVaR under model uncertainty}

In practical risk management, the true distribution of a financial position is rarely known precisely. 
It is therefore important to evaluate the worst-case risk over a plausible set of distributions. 
We consider two classical uncertainty sets: one defined by a Wasserstein distance and another described by moment constraints.

The Wasserstein metric is widely used in optimal transport and distributionally robust optimization. 
In this subsection, we first consider uncertainty sets induced by the Wasserstein distance --- referred to as Wasserstein uncertainty. 
For two random variables $X \sim F$ and $Y \sim G$, the Wasserstein distance of order $k \in [1,\oo)$ is defined as
$$
W_k(X, Y) = W_k(F, G) = \inf\{\left\| X-Y\right\|_k:\ X \sim F,\, Y \sim G\} = \bigg(\int_0^1 |F_X^{-1}(\alpha) - F_Y^{-1}(\alpha)|^k \d\alpha \bigg)^{1/k},
$$
where $F_X^{-1}$ and $F_Y^{-1}$ are the quantile functions of $X$ and $Y$, respectively.

\begin{proposition}\label{prop43}
Let $p \geq 1, \delta \geq 0$, and let $\Lambda: \R \to [0,1)$ be a decreasing function. 
For any $X \in L^p$, 
\begin{equation}\label{EVaR_Wass}
\sup\left\{\EVaR^p_\Lambda(Y): W_p(X,Y) \leq \delta \right\}
= \sup_{x \in \mathbb{R}}\left\{\( \EVaR^p_{\Lambda(x)}(X) + \delta (1-\Lambda(x))^{-1/p} \) \wedge x \right\}.
\end{equation}
\end{proposition}
\begin{proof}
From Proposition 1 of \cite{PPW12}, we have for every $\alpha \in [0,1)$
\begin{equation*}
	\sup\left\{\EVaR^p_\alpha(Y): W_p(X, Y) \leq \delta \right\}  = \EVaR^p_\alpha(X) + \delta (1-\alpha)^{-1/p}.
\end{equation*}
Combining with \eqref{eq-2026-1111}, we have
\begin{eqnarray*}
\sup\left\{\EVaR_\Lambda^p(Y): W_p(X, Y) \leq \delta \right\} &=& \sup\left\{\sup_{x \in \R} \left\{ \EVaR_{\Lambda(x)}^p(Y) \wedge x \right\}: W_p(X, Y) \leq \delta \right\}    \\
&=& \sup_{x \in \R} \left\{ \sup\left\{\EVaR_{\Lambda(x)}^p(Y): W_p(X, Y) \leq \delta \right\} \wedge x \right\}   \\
&=& \sup_{x \in \R} \left\{ \(\EVaR^p_{\Lambda(x)}(X) + \delta (1-\Lambda(x))^{-1/p} \) \wedge x \right\}.
\end{eqnarray*}
This completes the proof.
\end{proof}

\begin{remark}{\rm
The expression on the right-hand side of \eqref{EVaR_Wass} cannot in general be simplified to a $\EVaR^p_{\widetilde{\Lambda}}(X)$ with a single, fixed decreasing function $\widetilde{\Lambda}$ independent of $X$.  
The reason is that the additive term $\delta (1-\Lambda(x))^{-1/p}$ depends on $\Lambda(x)$ in a non-linear way, and the mapping $\alpha \mapsto \EVaR^p_\alpha(X) + \delta (1-\alpha)^{-1/p}$ does not, in general, coincide with $\EVaR^p_{\widetilde{\alpha}}(X)$ for a unique $\widetilde{\alpha}$ that is independent of the distribution of $X$.  
Therefore, \eqref{EVaR_Wass} is already the most compact representation of the worst-case $\Lambda$-EVaR under a Wasserstein-distance constraint.
}
\end{remark}

Next, we study the uncertainty set induced by mean and variance --- a classical and widely used setting in distributionally robust optimization. 
This framework assumes that only the first two moments of the random variable are known, reflecting many practical situations where full distributional information is unavailable. 
For $m \in \R$ and $v \geq 0$, denote by
$${\cal L}(m, v) := \left\{X \in L^2: \E[X] = m,\ \E[|X-m|^2] \le v^2\right\}.$$
\begin{proposition}\label{prop45}
Let $\Lambda: \R \to [0,1)$ be a decreasing function. 
For any $X \in L^2$, 
\begin{equation*}
\sup_{X \in {\cal L}(m, v)} \VaR_\Lambda(X) = \sup_{X \in {\cal L}(m, v)} \ES_\Lambda(X) = \sup_{X \in {\cal L}(m, v)} \EVaR^2_\Lambda(X) = \sup_{x \in \mathbb{R}}\left\{\(m + v \Lambda(x)^{1/2} (1-\Lambda(x))^{-1/2}\) \wedge x \right\}.
\end{equation*}
\end{proposition}
\begin{proof}
From Example 9 in \cite{PWW25}, when the mean $m$ and variance $v$ are specified, the Cantelli bounds for VaR and ES are
\begin{equation*}
\sup_{X \in {\cal L}^2(m, v)} \VaR_{\alpha}(X) = \sup_{X \in {\cal L}^2(m, v)} \ES_{\alpha}(X) =  m + v \alpha^{1/2}(1-\alpha)^{-1/2} ~~~ \alpha \in [0, 1).
\end{equation*}
Furthermore, Corollary 1 in \cite{Li18} readily yields
\begin{eqnarray*}
	\sup_{X \in {\cal L}(m, v)} \EVaR^2_{\alpha}(X) = m + v \alpha^{1/2}(1-\alpha)^{-1/2} ~~~ \alpha \in [0, 1).
\end{eqnarray*}
Consequently,
\begin{eqnarray*}
\sup_{Y \in {\cal L}(m, v)} \VaR_\Lambda(Y) &=& \sup_{Y \in {\cal L}(m, v)} \sup_{x \in \R} \left\{ \VaR_{\Lambda(x)}(Y) \wedge x \right\} = \sup_{x \in \R} \left\{ \sup_{Y \in {\cal L}(m, v)} \VaR_{\Lambda(x)}(Y) \wedge x \right\}   \\
&=& \sup_{x \in \R} \left\{ \(m + v \Lambda(x)^{1/2}(1-\Lambda(x))^{-1/2}\) \wedge x \right\}.
\end{eqnarray*}
The corresponding expressions for $\ES_\Lambda$ and $\EVaR^2_\Lambda$ are obtained analogously.
\end{proof}

\begin{remark}[Limitations on higher-moment constraints and general $p$]{\rm 
The elegant coincidence of worst-case $\Lambda$-VaR,  $\Lambda$-ES and  $\Lambda$-EVaR of order 2 under mean-variance constraints relies on two special features:
\begin{itemize}
\item Second-order nature of the Cantelli bound. The closed-form expression $m + v \Lambda(x)^{1/2}(1-\Lambda(x))^{-1/2}$ stems from the classical one-sided Chebyshev (Cantelli) inequality, which is intrinsically tied to the second moment. 
For constraints involving higher moments (e.g., skewness or kurtosis), no such universal bound exists, and the worst-case distribution typically depends on the specific risk measure being optimized.

\item The general $p$. For $p \neq 2$, the worst-case problem $\(\E[(X-t)_+^p]\)^{1/p}$ generally admits no explicit solution when only mean and higher-moment know.
\end{itemize}
}
\end{remark}

\section{Conslusion}\label{sec5}

This paper introduces and studies the Lambda extension of the R\'{e}nyi entropic value-at-risk ($\Lambda$-EVaR), a flexible family of risk measures that allows the confidence level to vary with the loss threshold via a decreasing function $\Lambda$. 

The main contributions are threefold. 
 
First, we have defined $\Lambda$-EVaR and established its fundamental properties, including monotonicity, cash subadditivity, quasi-convexity, and $L^p$-continuity. 
A key result (Proposition \ref{prop35}) shows that convexity, concavity in mixtures and cash additivity of $\Lambda$-EVaR are all equivalent to $\Lambda$ being constant, confirming that the Lambda framework genuinely extends the classical EVaR.  

Second, we have derived a dual representation (Theorem \ref{thm41}) and an extended Rockafellar-Uryasev formula (Theorem \ref{EVaR_US}) for $\Lambda$-EVaR. 
The latter provides a tractable two dimensional minimization problem that facilitates numerical computation and optimization.  

Third, we have examined worst-case values of $\Lambda$-EVaR under two standard models of distributional uncertainty: Wasserstein balls and mean-variance sets. 
For Wasserstein uncertainty, the worst-case $\Lambda$-EVaR admits a closed-form expression (Proposition \ref{prop43}) that cannot be reduced to a single $\Lambda$-EVaR with a new fixed $\Lambda$. 
Under mean-variance constraints, the worst-case values of $\Lambda$-VaR, $\Lambda$-ES and $\Lambda$-EVaR (with $p=2$) coincide and are given by a simple Cantelli-type bound (Proposition \ref{prop45}). 
This coincidence, however, does not extend to higher-order moment constraints or to general $p\neq2$.

The $\Lambda$-EVaR family thus offers a unified framework and incorporating the sensitivity to higher moments afforded by the R\'{e}nyi entropy. 
Future research could explore its statistical estimation, its performance in portfolio optimization under ambiguity, and its potential applications in regulatory capital frameworks where varying confidence levels are economically justified.

\section*{Acknowledgements}

Z. Zou is supported by National Natural Science Foundation of China (No. 12401625), and the Fundamental Research Funds for the Central Universities (No. WK2040000108).

\section*{Disclosure statement}

No potential conflict of interest was reported by the authors.

\end{document}